\documentclass[11pt]{article}
\def\final{1}

\usepackage[utf8]{inputenc}
\usepackage[T1]{fontenc}

\usepackage{natbib}
 \bibpunct[, ]{(}{)}{,}{a}{}{,}%

\usepackage{paralist,caption}
\usepackage{verbatim}
\usepackage{graphicx}
\usepackage{tikz}
\usepackage{pgf}
\usepackage{amsmath,amsfonts,amssymb,amsthm}
\usepackage{paralist}
\usepackage{xcolor,hyperref}
\usetikzlibrary{arrows,automata}
\usepackage{tkz-graph}
\GraphInit[vstyle = Shade]
\tikzset{
  LabelStyle/.style = { rectangle, rounded corners, draw,
                        minimum width = 2em, fill = yellow!50,
                        text = red, font = \bfseries },
  VertexStyle/.append style = { inner sep=5pt,
                                font = \Large\bfseries},
  EdgeStyle/.append style = {->, bend left} }
\usetikzlibrary{arrows}
\usetikzlibrary{shapes,snakes,calendar,matrix,backgrounds,folding}

\definecolor{darkblue}{rgb}{0,0,0.55}
\definecolor{darkred}{rgb}{0.6,0,0}
\definecolor{darkgreen}{rgb}{0.1,0.35,0}
\hypersetup{colorlinks, linkcolor=darkblue, citecolor=darkgreen, urlcolor=darkblue}

\newcommand{\Ei}{E}

\newcommand{\pulse}{\text{\sc pulse}}

\newcommand{\exponential}{\text{\sc exponential}}

\newcommand{\damper}{\text{\sc damper}}
\newcommand{\warmupc}{\alpha}
\newcommand{\plen}{\rho}
\newcommand{\dampreq}{Q}

\newcommand{\qon}{\plen}
\newcommand{\qoff}{\qon}

\tikzset{gadget/.style={fill=black!20, draw=black!60,text=black,inner sep=0.8cm,rounded corners, dashed}}

\ifdefined\final
\newcommand{\nnote}[1]{}
\newcommand{\jnote}[1]{}
\newcommand{\tnote}[1]{}

\newcommand{\todo}[1]{}

\else
\newcommand{\nnote}[1]{\begingroup \color{blue!60!black} \em Neil: #1 \endgroup}
\newcommand{\jnote}[1]{\begingroup \color{red} \em Jose: #1 \endgroup}
\newcommand{\todo}[1]{\begingroup \color{red!80!black} \em TODO: #1 \endgroup}

\fi

\newtheorem{theorem}{Theorem}

\newtheorem{lemma}{Lemma}

\theoremstyle{definition}

\newtheorem{example}{Example}

\usepackage{authblk}

\begin{document}

\title{Long term behavior of dynamic equilibria in fluid queuing networks}
\author{Roberto Cominetti}
\affil{Facultad de Ingenier\'ia y Ciencias, Universidad Adolfo Ib\'a\~nez}
\author{Jos\'e Correa}
\affil{Departamento de Ingenier\'ia Industrial, Universidad de Chile}
\author{Neil Olver}
\affil{Department of Mathematics, London School of Economics and Political Science}
\date{}

\maketitle

\begin{abstract}A fluid queuing network constitutes one of the simplest models in which to study flow dynamics over a network. In this model we have a single source-sink pair and each link has a per-time-unit capacity and a transit time. A dynamic equilibrium (or equilibrium flow over time) is a flow pattern over time such that no flow particle has incentives to unilaterally change its path.  Although the model has been around for almost fifty years, only recently results regarding existence and characterization of equilibria have been obtained. In particular the long term behavior remains poorly understood. Our main result in this paper is to show that, under a natural (and obviously necessary) condition on the queuing capacity, a dynamic equilibrium reaches a steady state (after which queue lengths remain constant) in finite time. Previously, it was not even known that queue lengths would remain bounded. The proof is based on the analysis of a rather non-obvious potential function that turns out to be monotone along the evolution of the equilibrium. Furthermore, we show that the steady state is characterized as an optimal solution of a certain linear program. When this program has a unique solution, which occurs generically, the long term behavior is completely predictable. On the contrary, if the linear program has multiple solutions the steady state is more difficult to identify as it depends on the whole temporal evolution of the equilibrium.
\end{abstract}

\section{Introduction}

The theory of flows over time provide a natural and convenient model to describe the dynamics of a continuous stream of particles traveling from a source to a sink in a network, such as urban or Internet traffic. 
Probably the most basic model for the propagation of flow is the so-called \emph{fluid-queue model} in which each arc in the network consists of a fluid queue with an arc-specific capacity followed by a link with constant delay. Thus, the time to traverse an edge is composed of a flow-dependent waiting time in the queue plus a constant travel time after leaving the queue. This model was initially studied in the framework of optimization. \cite{FyF:2,FyF:1} considered a fluid queue model in a discrete time setting and designed an algorithm to compute a flow over time carrying the maximum possible flow from the source $s$ to the sink $t$ in a given timespan. \cite{Gale:1} then showed the existence of a flow pattern that achieves this optimum simultaneously for all time horizons. These results were extended to continuous time by  \cite{Anderson:1} and \cite{Tardos:1}. We refer to \cite{Skutella:1} for an excellent survey. However, when network flows suffer from a lack of coordination among the participating agents, it is natural to take a game theoretic approach. As first described by \cite{Vickrey} for a simple \emph{bottleneck model}, in a dynamic network routing game each infinitesimal particle is interpreted as a player that seeks to complete its journey in the least possible time. Players are forward-looking and anticipate the congestion and queuing delays induced by others upon arrival to any edge in the network. Equilibrium occurs when each particle travels along a shortest path. 

More formally, a \emph{fluid queuing network} is a directed graph $G=(V,\Ei)$ where each arc $e\in\Ei$ consists of a fluid queue 
with capacity $\nu_e>0$  followed by a link with constant delay $\tau_e\geq 0$ (see Figure \ref{fig:arc1}).
A constant inflow rate $u_0>0$ enters the network at a fixed source $s\in V$ and travels towards a terminal node $t\in V$. 
A \emph{dynamic equilibrium} models the temporal evolution of the flows in the network.
Loosely speaking, it consists of a flow pattern in which every particle travels along a shortest path,
accounting for the fact that travel times depend on the instant at which a particle enters the network as well as 
the state of the queues that will be encountered along its path by the time at which they are reached.
Intuitively, if the queues are initially empty, the equilibrium should start by sending all the flow along 
shortest paths considering only the free-flow delays $\tau_e$. 
These paths are likely to become overloaded so that queues will grow on some of its edges and at some point in time 
new paths will become competitive and will be incorporated into the equilibrium. 
These new paths may in turn build queues so that even longer paths may
come into play. Hence one might expect that the equilibrium proceeds 
in phases in which  the paths used by the equilibrium remain stable. 
However, it is unclear if  the number of such phases is finite and whether the
equilibrium will eventually reach a steady state in which the queues and travel times
stabilize.
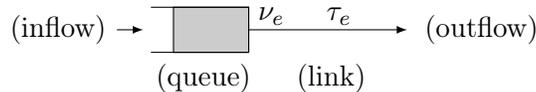
\begin{figure}[t]
        \centering
        \begin{tikzpicture}[>=latex]
        \filldraw[fill=black!20] (1.3,-0.55) rectangle (0.3,0.05);
        \draw (0.0,0.05) -- ++(1.3cm,0) -- ++(0,-0.6cm) -- ++(-1.3cm,0);
        \draw[<-] (-0.1,-0.25) -- +(-10pt,0) node[left] {(inflow)};
        \node at (4.4,-0.25cm) {(outflow)};

         \draw[->] (1.3,-0.25) -- +(60pt,0);
        \node at (1.6,-0.07cm) {$\nu_e$};
        \node at (2.5,-0.07cm) {$\tau_e$};
        \node at (0.7cm,-0.9cm) {(queue)};
        \node at (2.4cm,-0.9cm) {(link)};
        \end{tikzpicture}
        \caption{An arc in the fluid queuing network.}
        \label{fig:arc1}
\end{figure}

Although dynamic equilibria have been around for almost fifty years (see, e.g., \citep{FyF:2,fri,Gale:1,mer,mer1,pep,RanBoyce,Vickrey,XuF}), their existence has only been proved recently by \cite{Zhu} though in a somewhat different setting, and by \cite{meu} who gave the first existence result for a model that covers the case of fluid queuing networks. These proofs, however, rely heavily on functional analysis techniques and provide little intuition on the combinatorial structure of dynamic equilibria, their characterization, or feasible approaches to compute them. Substantial progress was recently achieved by \cite{KochSk:2} by introducing the concept of \emph{thin flows with resetting} that characterize the time derivatives of a dynamic equilibrium, and which provide in turn a method to compute an equilibrium by integration. A slightly refined notion of \emph{normalized} thin flows with resetting was considered by \cite{ccl}, who proved existence and uniqueness, and provided a constructive proof for the existence of a dynamic equilibrium. 

In recent work, further extensions and variants of the model have been studied. In particular, \cite{SS18} extend some of the known results about dynamic equilibria to the case in which there are multiple sources and terminal nodes. However, the multi-commodity case is largely open. Furthermore, \cite{SV19} consider spill back effects to model the fact that in practice the queues cannot grow arbitrarily large and that their effect propagates back in the network. \cite{GH19} consider a related model in which particles behave myopically and make routing decisions based on the current state of the network, without anticipating its evolution. Finally, we mention the work of \cite{C17} who considered an atomic model and established that in series-parallel networks queues remain bounded in a dynamic equilibrium.

In this paper we focus on the long term behavior of dynamic equilibria in fluid queuing networks.  
Clearly if the inflow $u_0$ is very large compared to the queuing capacities, the queues will 
grow without bound, and no steady state can be expected. 
More precisely, let  $\delta^+(S)$ be an $st$-cut
with minimum queuing capacity $\bar \nu=\sum_{e\in \delta^+(S)}\nu_e$; 
if there are multiple options, choose $S$ (containing $s$) to be setwise minimal.
If $u_0>\bar\nu$ 
all the arcs in $\delta^+(S)$ will grow unbounded queues, whereas for $u_0\leq\bar\nu$, it is natural to expect that 
the equilibrium should eventually reach a steady state, where queue lengths remain constant.
This was not known---in fact, it was not even known that queue lengths remain bounded!

Our main goal in this paper is to show that both these properties do indeed hold: more precisely, when $u_0\leq\bar\nu$, the dynamic equilibrium reaches a steady state in \emph{finite} time. At first glance, these convergence properties might seem ``obvious'', and it might seem surprising that they are at all difficult to prove. We will present some examples that illustrate why this is not the case. For instance, it may occur that the flow across the cut $\delta^+(S)$ may temporarily exceed its
capacity $\bar\nu$ by an arbitrarily large factor, forcing the queues to grow very large. This phenomenon may
occur since the inflow $u_0$ entering the network at different points in time may experience 
different delays and eventually superpose at $\delta^+(S)$ which gets an inflow larger than $u_0$. 
In other cases some queues may grow during a period of time after which 
they reduce to zero and then grow again later on. In fact, we give a construction that shows that this can happen an exponential (in the input size) number of times during the evolution! Along the way to our main result, we provide a characterization of the steady state as an optimal solution of a certain
linear programming problem and we discuss when this problem has a unique solution. 
Despite the fact  that convergence to a steady state occurs in finite time, it remains as an open 
question whether this state is attained after finitely many phases. 

The paper is structured as follows. Section~\ref{review} reviews the model of 
fluid queuing networks, including the precise definition of dynamic equilibrium and the main 
results known so far. Then, in Section~\ref{steady} we discuss the notion of steady state and provide a characterization in terms
of a linear program. Inspired by the objective function of this linear program, in Section~\ref{potential} we 
introduce a potential function and we prove that it is a Lyapunov function for the dynamics. This 
potential turns out to be piecewise linear in time with finitely many possible slopes. 
We then prove that the potential remains bounded so that there is a finite time at which
its slope is zero, and we show that in that case the system has reached a steady state.
Further, we provide an explicit pseudopolynomial bound on the convergence time.
Finally, in Section~\ref{examples} we discuss the interesting (and perhaps surprising) examples alluded to earlier,
as well as some remaining open questions.

\section{Dynamic equilibria in fluid queuing networks}\label{review}
In this section we recall the definition of dynamic equilibria in fluid queuing networks, 
and  we briefly review the known results on their existence, characterization, and computation.
The results are stated without proofs for  which we refer to \cite{KochSk:2} and \cite{ccl}.

\subsection{The model}
Consider a fluid queuing network $G=(V,E)$ with arc capacities $\nu_e$ and delays $\tau_e$.
The network dynamics are described in terms of the inflow rates $f_e^+(\theta)$ that enter each 
arc $e\in \Ei$ at time $\theta$, where $f_e^+:[0,\infty)\to[0,\infty)$ is measurable.

\medskip

\noindent{\bf Arc dynamics.} If the inflow $f_e^+(\theta)$ exceeds $\nu_e$ a queue $z_e(\theta)$
will grow at the entrance of the arc. 
\begin{figure}[h]
        \centering
        \begin{tikzpicture}[>=latex]
        \filldraw[fill=black!20] (1.3,-0.55) rectangle (0.3,0.05);
        \draw (0.0,0.05) -- ++(1.3cm,0) -- ++(0,-0.6cm) -- ++(-1.3cm,0);

         \draw[->] (1.3,-0.25) -- +(60pt,0);
        \draw[<-] (-0.1,-0.25) -- +(-10pt,0) node[left] {$f_e^+(\theta)$};
        \node at (4.2,-0.25cm) {$f_e^-(\theta\!+\!\tau_e)$};
        \node at (0.81,-0.25cm) {$z_e(\theta)$};
        \node at (1.55,-0.1cm) {$\nu_e$};
        \node at (2.5,-0.1cm) {$\tau_e$};
        \node at (-1.1cm,-0.8cm) {(inflow)};
        \node at (0.7cm,-0.8cm) {(queue)};
        \node at (2.4cm,-0.8cm) {(link)};
        \node at (4.2cm,-0.8cm) {(outflow)};
        \end{tikzpicture}
        \caption{Dynamics of an arc in the queuing network.}
        \label{fig:arc}
\end{figure}
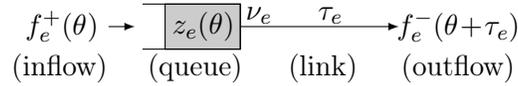
The queues are assumed to operate at capacity, that is to say,  when $z_e(\theta)>0$ the flow is released at rate $\nu_e$, whereas
when the queue is empty the outflow is the minimum between $f_e^+(\theta)$ and the capacity $\nu_e$. 
Hence the queue evolves from its initial state $z_e(0)\!=\!0$ according to
\begin{equation}\label{zprime}
\dot z_e(\theta)=\left\{\begin{array}{cl}
f_e^+(\theta)-\nu_e\hspace{1.2ex}&\mbox{if }z_e(\theta)>0\\
\left[f_e^+(\theta)-\nu_e\right]_+&\mbox{if }z_e(\theta)=0.
\end{array}\right.
\end{equation}
These dynamics uniquely determine the queue lengths $z_e(\theta)$ as well as the arc
outflows
\begin{equation}\label{outflow}
f^-_e(\theta\!+\!\tau_e)=\left\{
\begin{array}{cl}
\nu_e&\mbox{if }z_e(\theta)>0\\
\min\{f_e^+(\theta),\nu_e\}&\mbox{if }z_e(\theta)=0.
\end{array}
\right.
\end{equation}

\medskip
\noindent{\bf Flow conservation.} 
A \emph{flow over time} is a family $(f_e^+)_{e\in\Ei}$ of arc inflows 
such that flow is conserved  at every node $v\in V\setminus\{t\}$, namely for a.e. $\theta\geq 0$
\begin{equation}\label{flowcons}
\sum_{e\in\delta^+(v)}\!\!\!f_e^+(\theta)\;-\!\!\!\sum_{e\in\delta^-(v)}\!\!\!f_e^-(\theta)=\left\{\begin{array}{cl}
u_0&\mbox{if }v=s\\
0&\mbox{if }v\neq s,t.
\end{array}\right.
\end{equation}

\medskip
\noindent{\bf Dynamic shortest paths.} 
A particle entering an arc $e$ at time $\theta$ experiences a queuing delay $z_e(\theta)/\nu_e$ plus
a free-flow delay $\tau_e$ to traverse the arc after leaving the queue, so that it will exit the arc at time
\begin{equation}\label{T}
T_e(\theta)=\theta+\frac{z_e(\theta)}{\nu_e}+\tau_e.
\end{equation}
Consider a particle entering the source node $s$ at time $\theta$. If this particle follows a path 
$p=e_1e_2\cdots e_k$, it will reach the end of the path at time
\begin{equation}\label{Tpath}
T_p(\theta)=T_{e_k}\circ\cdots\circ T_{e_{2}}\circ T_{e_1}(\theta).
\end{equation}
Denoting $\mathcal P_v$ the set of all $sv$-paths, the minimal time at which node $v$ can be reached is
\begin{equation}\label{labels}
\ell_v(\theta)=\min_{p\in\mathcal P_v}T_p(\theta).
\end{equation}
The paths attaining these minima are called  {\em dynamic shortest paths}. The
arcs in these paths are said to be {\em active} at time $\theta$ and
we denote them by $E'_\theta$. Observe that $\ell_v(\theta)$ can also be 
defined through the dynamic Bellman's equations
\begin{equation}\label{bellman}
\left\{\begin{array}{l}
\ell_s(\theta)=\theta\\[1ex]
\ell_w(\theta)={\displaystyle \min_{e=vw\in E}}T_e(\ell_v(\theta))
\end{array}\right.
\end{equation}
so that $e=vw$ is active precisely if $\ell_w(\theta)=T_e(\ell_v(\theta))$.

\medskip
\noindent{\bf Dynamic equilibrium.}  
A dynamic equilibrium is a flow pattern that uses only dynamic shortest paths. More precisely,
let $\Theta_e=\{\theta:e\in E_\theta'\}$ be the set of \emph{entrance times} $\theta$ at which 
the arc $e$ is active, and $\Xi_e=\ell_v(\Theta_e)$ the set 
of \emph{local times} $\xi=\ell_v(\theta)$ at which $e$ will be active. 
A flow over time $(f_e^+)_{e\in\Ei}$ is called a \emph{dynamic equilibrium} iff
for almost every $\xi\geq 0$ we have $f_e^+(\xi)>0\Rightarrow \xi\in\Xi_e$.

\subsection{Characterization of dynamic equilibria}
Since the inflows $f_e^+(\cdot)$ are measurable the same holds for $f_e^-(\cdot)$ 
and we may define the \emph{cumulative inflows} and \emph{cumulative outflows} as
\begin{eqnarray*}
F_e^+(\theta)&=&\mbox{$\int_0^\theta f_e^+(z)\,dz$}\\
F_e^-(\theta)&=&\mbox{$\int_0^\theta f_e^-(z)\,dz$}.
\end{eqnarray*}
These cumulative flows allow to express the queues as $z_e(\theta)=F_e^+(\theta)-F_e^-(\theta+\tau_e)$.
It turns out that a dynamic equilibrium can be equivalently characterized by the fact that for 
each arc $e=vw\in E$ we have 
\begin{equation}\label{dyneq}
F_e^+(\ell_v(\theta))=F_e^-(\ell_w(\theta))\quad\forall\;\theta\geq 0.
\end{equation}
In this case, the functions $x_e(\theta)\triangleq F_e^+(\ell_v(\theta))$ are
static flows with
\begin{equation}\label{staticflow}
\sum_{e\in\delta^+(v)}x_e(\theta)-\sum_{e\in\delta^-(v)}x_e(\theta)=\left\{\begin{array}{rl}
u_0\theta&\mbox{ if }v=s\\
-u_0\theta&\mbox{ if }v=t\\
0&\mbox{ if }v\neq s,t.
\end{array}\right.
\end{equation}

\subsection{Derivatives of a dynamic equilibrium}

The labels $\ell_v(\theta)$ and the static flows $x_e(\theta)$ are nondecreasing functions 
which are also
absolutely continuous so that they can be reconstructed from their derivatives by integration.\footnote{These 
derivatives exist almost everywhere and are locally integrable.} Moreover, from these functions 
one can recover the equilibrium inflows $f_e^+(\cdot)$ using the relation $x_e'(\theta)=f_e^+(\ell_v(\theta))\ell_v'(\theta)$.
Hence, finding a dynamic equilibrium reduces essentially to computing the derivatives $\ell'_v(\theta), x'_e(\theta)$.

Let $\theta$ be a point of differentiability and set $\ell_v'=\ell'_v(\theta)\geq 0$ and $x_e'=x'_e(\theta)\geq 0$.
From \eqref{staticflow} we see that $x'$ is a static $st$-flow of size $u_0$, namely, 
\begin{equation}\label{staticflowder}
\sum_{e\in\delta^+(v)}x_e'-\sum_{e\in\delta^-(v)}x_e'=\left\{\begin{array}{rl}
u_0&\mbox{ if }v=s\\
-u_0&\mbox{ if }v=s\\
0&\mbox{ if }v\neq s,t
\end{array}\right.
\end{equation}
while using \eqref{bellman}, \eqref{T}, \eqref{zprime} and the differentiation rule for a minimum 
we get
\begin{equation}\label{lprime}
\left\{\begin{array}{l}
\ell'_s=1\\[1ex]
\ell_w'={\displaystyle \min_{e=vw\in E_\theta'}}\rho_e(\ell'_v,x_e')
\end{array}\right.
\end{equation}
where
\begin{equation}\label{rhoe}
\rho_e(\ell'_v,x_e')=\left\{\begin{array}{cl}
x_e'/\nu_e&\mbox{if }e\in E_\theta^*\\
\max\{\ell'_v,x_e'/\nu_e\}&\mbox{if }e\not\in E_\theta^*
\end{array}\right.
\end{equation}
with $E_\theta^*$ the set of arcs $e=vw$ with positive queue $z_e(\ell_v(\theta))>0$. 
In addition to this, the conditions for dynamic equilibria imply $E_\theta^*\subset E_\theta'$
as well as 
\begin{equation}\label{ntfr}
\begin{array}{ll}
(\forall\,e\in E_\theta') \qquad~ &x_e'>0\;\Rightarrow\; \ell_w'=\rho_e(\ell_v',x_e')\\
(\forall\,e\not\in E_\theta') \qquad~ &x_e'=0.
\end{array}
\end{equation}

These equations fully characterize the derivatives of a dynamic equilibrium. 
In fact, for all subsets $E^*\subseteq E'\subseteq E$ the system 
\eqref{staticflowder}-\eqref{ntfr} admits at least one solution $(\ell',x')$ and moreover the $\ell'$ 
component is unique. These solutions are called {\em normalized thin flows with 
resetting} {\sc (ntfr)} and can be used to reconstruct a dynamic equilibrium by integration,
proving the existence of equilibria.
We refer to \cite{ccl} for the existence and uniqueness of {\sc ntfr}'s and to \cite{KochSk:2} 
for a description of the integration algorithm and how to find the equilibrium inflows 
$f_e^+(\cdot)$.

Observe that there are only finitely many options for $E^*$ and $E'$. Since the 
corresponding $\ell'$ is unique, it follows that the functions $\ell_v(\theta)$
will be uniquely defined and piecewise linear with finitely many options for
the derivatives. Although the static flows 
$x_e(\theta)$ are not unique in general, one can still find an equilibrium
in which these functions are also piecewise linear by fixing
a specific $x'$ in the {\sc ntfr} for each pair $E^*,E'$.

\subsection{A detailed example}

We now work out the details of a small example that already provides some intuition on how the dynamic equilibria behaves. In particular this example exhibits an unexpected property, namely that the flow coming into the sink can be larger than the network inflow. 

\begin{example}\label{eg:1}
Consider the network consisting of the vertices $\{s,v,t\}$ with edges $e=(s,t),f=(s,v),g=(v,t),h=(v,t)$ and inflow $u_0=u$. Capacities are $\nu_e=u/3$, $\nu_f=3u/4$, $\nu_g=u/3$, and $\nu_h=u$, and delays are $\tau_e=\tau_h=\tau$, and $\tau_f=\tau_g=0$. 
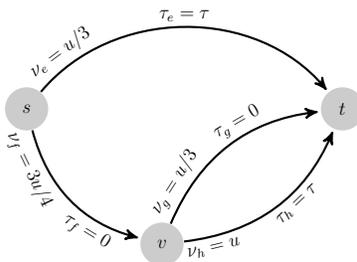
\begin{figure}[h]
    \centering
\begin{tikzpicture}[->,>=stealth',shorten >=1pt,auto,node distance=2.8cm, thick, scale=0.6, every node/.style={scale=0.6}]
\tikzstyle{every state}=[fill=black!20,draw=none,text=black]

  \node[state] (s)   at (0,0)                 {\large $s$};
  \node[state]         (t) at (7,0) {\large $t$};
  \node[state]         (v) at (3,-3) {\large $v$};

  \path (s) edge   [bend left=50]   node [pos=0.12, above, sloped]  {\large $\nu_e=u/3$} 
   						  node [pos=0.5, above, sloped]  {\large $\tau_e=\tau$} (t)
                  edge   [bend right] 	  node  [pos=0.2, below, sloped]  {\large $\nu_f=3u/4$} 
            					  node  [pos=0.7, below, sloped]  {\large $\tau_f=0$}  (v)
            (v) edge   [bend left] 	  node  [pos=0.18, above, sloped]  {\large $\nu_g=u/3$}
        						  node  [pos=0.6, above, sloped]  {\large $\tau_g=0$}  (t)
                 edge   [bend right]        node  [pos=0.13, below, sloped]  {\large $\nu_h=u$}
            					  node  [pos=0.6, below, sloped]  {\large $\tau_h=\tau$} (t);
\end{tikzpicture}
\end{figure}

In a dynamic equilibrium for this instance, flow is initially routed through the shortest path $fg$. 
Then queues grow in both edges until at time $\tau/2$ the path consisting of $e$ enters the shortest path network. 
From that point in time the flow splits in equal proportions between paths $e$ and $fg$ implying that a queue starts growing on edge $e$, 
the queue of $f$ starts decreasing, while the queue on $g$ continues to increase. 
Then, at time $0.7\tau$, the queue on $g$ has grown enough to make $h$ enter the shortest path network. 
At this point all edges except $h$ have a queue and the flow starts splitting as follows: 
$4u/13$ take path $e$, $4u/13$ take path $fg$, and $5u/13$ take path $fh$. 
Therefore all queues start decreasing until at time $2\tau$ the queues on edges $e$ and $f$ deplete simultaneously. 
When this happens the shortest path network stays the same but the edges with queue change since only $g$ still has a queue. 
The new thin flow is thus computed, and the flow starts splitting evenly between the paths $e$, $fg$, and $fh$ (each gets flow $u/3$). 
This last phase constitute a steady state so it lasts forever. 
More precisely in this instance one can compute the derivative of the distance labels at node $t$ as
\begin{equation}\nonumber
\ell'_t(\theta)=\left\{\begin{array}{cl}
3&\mbox{ for }\theta\in [0,\tau/2)\\
3/2&\mbox{ for }\theta\in [\tau/2,\tau/2+\tau/5)\\
12/13&\mbox{ for }\theta\in [\tau/2+\tau/5,2\tau)\\
1&\mbox{ for }\theta\in [2\tau,\infty)\\
\end{array}\right. .
\end{equation}

Interestingly, the amount of flow arriving at $t$ at time $\ell_t(\theta)$ can readily be computed as $u/\ell'_t(\theta)$. So that if we consider the local time at node $t$ this flow is then 
\begin{equation}\nonumber
f^-_e(\theta)+f^-_g(\theta)+f^-_h(\theta)=\left\{\begin{array}{cl}
u/3&\mbox{ for }\theta\in [0,3\tau/2)\\
2u/3&\mbox{ for }\theta\in [3\tau/2,9\tau/5)\\
13u/12&\mbox{ for }\theta\in [9\tau/5,3\tau)\\
u&\mbox{ for }\theta\in [3\tau,\infty).\\
\end{array}\right.
\end{equation}
This brings us to the surprising fact that for some time interval, the flow arriving at the sink is larger than the inflow.
\end{example}

\section{Steady states}\label{steady}
We say that a dynamic equilibrium attains a \emph{steady state}  if for sufficiently large times all the queues are frozen to a constant 
$z_e(\theta)\equiv z^*_e$. This is clearly equivalent to the fact that  the arc travel 
times become constant equal to $\tau_e^*=\tau_e+q_e^*$ with $q_e^*=z_e^*/\nu_e$ the corresponding queuing times. 

\begin{lemma}\label{L1}
A dynamic equilibrium attains a steady state iff  there exists some $\theta^*\geq 0$ such that 
$\ell_v'(\theta)=1$ for every node $v\in V$ and all $\theta\geq\theta^*$. 
\end{lemma}
\begin{proof}
In a steady state we clearly have 
$\ell_v(\theta)=\theta+d^*_v$ where  $d_v^*$ is the minimum 
travel time from $s$ to $v$ with arc times $\tau^*_e$, so that $\ell_v'(\theta)=1$.
Conversely, if all these derivatives are equal to 1 then $\ell_v(\theta)=\theta+d^*_v$ 
for some constant $d_v^*$ and $\theta\geq\theta^*$. Moreover, an arc $e=vw$ with nonempty queue
must be active so that $ \ell_w(\theta)=T_e(\ell_v(\theta))$ which yields
$$z_e(\theta+d_v^*)=z_e(\ell_v(\theta))=\nu_e(\ell_w(\theta)-\ell_v(\theta)-\tau_e)=\nu_e(d^*_w-d^*_v-\tau_e)$$
which shows that all queues eventually become constant.
\end{proof}

\begin{theorem}
Consider a steady state with queues $z_e^*\geq 0$ and let $d_v^*$ be the minimum 
travel time from $s$ to $v$ under arc travel times $\tau_e^*=\tau_e+q_e^*$, where $q_e^*=z_e^*/\nu_e$. Let $(\ell',x')$ with $\ell'_v=1$ for all $v \in V$ be
a corresponding {\sc ntfr} and denote by $\mathcal F_0$  the set of $st$-flows of value $u_0$.
Then $x'$ and $(d^*,q^*)$ are optimal solutions to the following pair of 
dual linear programs:
\begin{equation}\tag{P}\label{eq:primal}
    \begin{aligned}
\min_{y'} \quad & \sum_{e \in \Ei} \tau_e y'_e \\
	\text{s.t.} \qquad y' &\in{\mathcal F}_0\\
	0 &\leq y'_e \leq \nu_e \qquad\forall e \in \Ei,
    \end{aligned}
\end{equation}
\begin{equation}\tag{D}\label{eq:dual}
    \begin{aligned}
\max_{d,q} \quad &u_0d_t -\sum_{e \in \Ei} \nu_e q_e\\
      \text{s.t.} \qquad  d_s&= 0\\
      d_w &\leq d_v+ \tau_e + q_e \qquad \forall e=vw \in \Ei\\
      q_e &\geq 0 \qquad\qquad\qquad\;\; \forall e\in \Ei.
  \end{aligned}
\end{equation}

\end{theorem}
\begin{proof}
Clearly $(d^*,q^*)$ is feasible for \eqref{eq:dual}. Also \eqref{staticflowder} 
gives $x'\in\mathcal F_0$, while  \eqref{ntfr} implies that if $x_e'>0$ then 
$1=\rho_e(1,x_e')$.
This implies that $x_e'\leq\nu_e$, so $x'$ is feasible for \eqref{eq:primal}.
If $x_e'>0$ then certainly the arc is active (formally, by \eqref{ntfr})
and hence $d_w^*=d_v^*+\tau_e+q^*_e$.
And if $q_e^*>0$, also implying that $e$ is active, then \eqref{lprime} implies that $1\leq \rho_e(1,x_e')=x_e'/\nu_e$, which yields $x_e'=\nu_e$. 
This proves that $x'$ and $(d^*,q^*)$ are complementary solutions, and hence are optimal for \eqref{eq:primal} and \eqref{eq:dual} respectively.
\end{proof}

According to this result, if a dynamic equilibrium eventually settles to a steady state
then the corresponding queue lengths must be optimal for \eqref{eq:dual}. Generically 
(after perturbing capacities) this linear program 
has a unique solution in which case the steady state is fully characterized. 
Otherwise, if  \eqref{eq:dual} has multiple solutions it is not evident which queue lengths 
will be obtained in steady state.
Note that even if the min cost flow for \eqref{eq:primal} is unique, this does not 
mean that only one steady state situation is possible because there 
may be flexibility in the queue lengths. For instance, if  $u_0=1$ and the network has 
a single link from $s$ to $t$ of unit capacity, if we create a queue of 
some length at time 0 this queue will remain in the steady state solution.
This point will be further discussed in Section~\ref{examples-ssq}.

\noindent{\bf Remark.}
It is not difficult to show that when we start with initial 
conditions $z_e(0)=z_e^*$ where $z_e^*=\nu_e q_e^*$ with $q^*$ optimal 
for \eqref{eq:dual}, then the dynamic equilibrium is already at a steady state and the 
queues remain constant. 

\section{Convergence to a steady state}\label{potential}

In this section we prove that a steady state exists and that it is actually reached in finite time.
To this end we introduce a Lyapunov potential function that increases along the evolution of the 
dynamic equilibrium. The potential function is inspired from the previous dual program and is given by
\[
    \Phi(\theta) := u_0(\ell_t(\theta)-\ell_s(\theta))-\sum_{e \in \Ei} z_e(\ell_v(\theta)).
\]

{\noindent \bf Remark.}
The potential is the difference between the total travel time experienced by users leaving at time $\theta$,
and the total queue volumes, as seen by users leaving at time $\theta$.
We are not aware of a more insightful interpretation of it.
In fact, none of the more ``natural'' quantities we tried as candidate potential functions (total delay, time spent queueing, total delay excluding queueing delays, \ldots) are monotone.

\begin{theorem}\label{phi_monotone}
    For every $\theta$ that is a point of differentiability of $\Phi$, $\Phi'(\theta)$ is nonnegative, and strictly positive unless the dynamic equilibrium has reached a steady state.
\end{theorem}
\begin{proof}
The queues can be expressed as 
$z_e(\ell_v(\theta))=\nu_e\left[\ell_w(\theta)-\ell_v(\theta)-\tau_e\right]_+$, and therefore 
$\Phi(\theta) = u_0(\ell_t(\theta)-\ell_s(\theta))-\sum_{e \in \Ei} \nu_e\left[\ell_w(\theta)-\ell_v(\theta)-\tau_e\right]_+$. To take the derivative we recall that $E_\theta'$ is the set of active edges, i.e., those for which 
$\ell_w(\theta)-\ell_v(\theta) \ge \tau_e$, while in $E_\theta^*$ the inequality is strict. 
Using the derivative of a max function and taking 
a {\sc ntfr}  $(\ell',x')$ at time $\theta$, we thus obtain
$$\Phi'(\theta) = u_0(\ell_t'-\ell_s')\,-\!\!\!\sum_{e \in E_\theta'\setminus E_\theta^*} \!\!\!\!\nu_e[\ell_w'-\ell_v']_+\, - \sum_{e \in E_\theta^*} \!\nu_e(\ell_w'-\ell_v').$$
Notice that the dependency of the $\tau_e$'s in the previous derivative is somewhat hidden in the set of active edges $E_\theta'$. Now, for $e\in E_\theta'\!\setminus\! E_\theta^*$ we have $\ell_w'\leq \rho_e(\ell_v',x_e')=\ell_v'$ if $x_e'=0$ and
$\ell_w'= \rho_e(\ell_v',x_e')\geq \ell_v'$ if $x_e'>0$, so that letting 
$E_\theta^+=E_\theta^*\cup \{e\in E_\theta'\setminus E_\theta^*:x_e'>0\}$ we may write
$$\Phi'(\theta) = u_0(\ell_t'-\ell_s')-\sum_{e \in E_\theta^+} \nu_e(\ell_w'-\ell_v').$$
Let us introduce a return arc $ts$ with capacity $\nu_{ts}=u_0$ and flow $x_{ts}'=u_0$ so that $x'$ is a circulation.
Let $E_\theta^r=E_\theta^+\cup\{ts\}$ and for each $e=vw\in E_\theta^r$ define the function
\[  H_e(z)=\begin{cases} 1&\text{ if }\ell_v'\leq z<\ell_w'\\
-1&\text{ if }\ell_w'\leq z<\ell_v'\\
0&\text{ otherwise.}\end{cases}\]
Then the derivative $\Phi'(\theta)$ can be expressed as
\[ \Phi'(\theta) = -\int_0^\infty\!\! \sum_{e \in E_\theta^r}\!\nu_e H_e(z)\,dz.\]

For the remainder of the proof, let $\delta^+(S)$ (respectively $\delta^-(S))$ denote the arcs in $E_\theta^r$ leaving (respectively entering) $S$. Let $V_z=\{v:\ell_v'\leq z\}$ and consider an arc $e =vw\in E_\theta^+$. 
If $e\in\delta^+(V_z)$ then $\ell_v'\leq z<\ell_w'$ and therefore $\ell_w'=x_e'/\nu_e$. Similarly, if 
$e\in\delta^-(V_z)$ then $\ell_w'\leq z<\ell_v'$ which implies $e\in E_\theta^*$ 
and again $\ell_w'=x_e'/\nu_e$. Hence $x_e'=\nu_e\ell_w'$ for all $e\in E_\theta^+\cap\delta(V_z)$.
This equality also holds for the return arc $ts$, while in the remaining arcs $x_e'=0$.
Hence 
\begin{equation}\label{eq:moneq}
    \sum_{e\in\delta^+(V_z)}\!\!\!\!\nu_e z\leq \!\!\!\!\!\sum_{e=vw\in\delta^+(V_z)}\!\!\!\!\!\nu_e \ell_w' = \!\!\!\sum_{e\in\delta^+(V_z)}\!\!\!\!x_e'
=\!\!\!\sum_{e\in\delta^-(V_z)}\!\!\!\!x_e'=\!\!\!\!\!\sum_{e=vw\in\delta^-(V_z)}\!\!\!\!\!\nu_e \ell_w' \leq \!\!\!\sum_{e\in\delta^-(V_z)}\!\!\!\!\nu_e z
\end{equation}
with strict inequality if $\delta^+(V_z)$ is nonempty.
It follows that for all $z > 0$ we have
\[\sum_{e \in E_\theta^r} \nu_e H_e(z)=\sum_{e\in\delta^+(V_z)}\nu_e -\sum_{e\in\delta^-(V_z)}\nu_e \leq0 \]
and therefore $\Phi'(\theta)\geq 0$ with strict inequality unless $\delta^+(V_z)$ is empty for almost all $z\geq 0$. 
Since for $\delta^+(V_z)$ to be empty we need that it either contains all vertices in $V$ or none of them, we have that $\Phi'(\theta) = 0$ if and only if all $\ell_v'$ are equal, and hence (since $\ell'_s = 1$) all equal to $1$. By Lemma~\ref{L1}, this exactly characterizes a steady state.
\end{proof}

\begin{theorem}
Let $\bar \nu=\sum_{e\in C}\nu_e$ be the minimal queuing capacity among all $st$-cuts $C$. If $u_0\leq\bar \nu$ then 
the dynamic equilibrium attains a steady state in finite time.
\end{theorem}
\begin{proof}
From Theorem \ref{phi_monotone} it follows that there is some $\kappa > 0$ such that $\Phi' (\theta)\geq \kappa$ for every phase other than the steady state. This is simply because the thin flow depends only on the current shortest path network $E_\theta'$ and the set of queuing edges $E_\theta^*$, and so there are only finitely many possible derivatives.

Thus, in order to prove that a steady state is reached in finite time it suffices to show that $\Phi(\theta)$ 
remains bounded. To this end we note that the condition $u_0\leq\bar \nu$ implies that \eqref{eq:primal} is feasible and hence it has a 
finite optimal value $\alpha$. The conclusion then follows by noting that the point $(d,q)$ with 
$d_v=\ell_v(\theta)-\ell_s(\theta)$ and $q_e=z_e(\ell_v(\theta))/\nu_e$ is feasible for the dual \eqref{eq:dual} so that $\Phi(\theta)\leq \alpha$.
\end{proof}

Given that convergence to a steady state does happen in finite time, it is natural to ask for explicit bounds. It is easy to see that a polynomial bound (in the input size encoding) is impossible; 
simply consider a network consisting of two parallel links, one with capacity $1-2^{-L}$ and length zero, the other with capacity $1$ and length $1$.
The first phase, where all traffic takes the shorter edge, lasts until time  $2^L-1$.
However, we can give a \emph{pseudopolynomial} bound on the convergence time (and hence, queue lengths).
The following results shows this bound. We present it in a slightly more general setting that allows for rational inflow and capacities and arbitrary initial queues since we will need it in this form in Section \ref{sec:exponential}.

\begin{theorem}\label{thm:pseudo}
    Consider an instance for which all arc capacities $\nu_e$ as well as the inflow $u_0$ are multiples of $1/K$, $K \in \mathbb{Z}_+$. 
    We allow for an arbitrary initial state at time $0$ with possibly nonempty queues.
    Let $M = \sum_{e \in E} \nu_e$ and $T = \sum_{e \in E} (\tau_e + q_e(0))$.
    Then assuming the dynamic equilibrium attains a steady state, it is reached by time $2K^2M^2T$, and moreover, 
    the waiting time in any queue never exceeds $2u_0K^3M^2T$. 
\end{theorem}
\begin{proof} 
We first remark that it suffices to prove the result for $K=1$, i.e., integer capacities and inflow.
For consider the instance where the inflow as well as all arc capacities have been scaled up by a factor $K$.
The equilibrium flow on this new instance is obtained by scaling the equilibrium flow on the original one, and so the time to reach steady state, as well as the queue waiting times at any moment in time, are the same for both instances.
Considering the impact of the scaling on the claimed bounds, the claim on this new instance thus implies the claim on the original instance, and so 
we assume $K=1$ for the remainder.

We use the same notions defined in the proof of Theorem~\ref{phi_monotone}.
Assume that a steady state is attained; thus \eqref{eq:primal} has a finite objective value, and this is at most $\sum_{e \in E} \nu_e\tau_e \leq M \sum_{e \in E} \tau_e$.
Thus $\Phi(\theta) \leq M\sum_{e \in E} \tau_e$ for all $\theta$.
    Initially, 
    \[ 
        \Phi(0) = u_0(\ell_t(0) - \ell_s(0)) - \sum_{e \in E} \nu_e q_e(0) \geq -M \sum_{e \in E} q_e(0).
    \]
    So $\Phi(\theta) - \Phi(0) \leq MT$ for all $\theta$.
    
    Consider some time $\theta$ which is a point of differentiability not in the steady state phase; so $\Phi'(\theta) > 0$ by Theorem~\ref{phi_monotone}.
    Our first goal will be to show that $\Phi'(\theta) \geq 1/(2M)$; this clearly implies the bound on the time to reach steady state.

    Let $\theta$ be any time before steady state is reached and for which $\Phi'(\theta)$ is defined.
    Let $z_1 = \min_{v \in V} \ell'_v(\theta)$ and $z_2 = \max_{v \in V} \ell'_v(\theta)$.
    From the proof of Theorem~\ref{phi_monotone}, we have $\sum_{e \in E_\theta^r} \nu_e H_e(z) \leq -1$ for any $z \in (z_1, z_2)$ (it is strictly negative and integral).
    And certainly $H_e(z) = 0$ for all $e \in E_\theta^r$ and $z \notin [z_1, z_2]$.
    Thus
    \[ \Phi'(\theta) = -\int_{z_1}^{z_2} \sum_{e \in E_\theta^r} \nu_e H_e(z)dz 
               \geq z_2 - z_1.
       \]
       To bound this, choose an arbitrary $z$ for which $z_1 < z < z_2$ (recall that $z_1 < z_2$ since we have not reached steady state).
    We have, following the lines of \eqref{eq:moneq},
    \[ 
        z_1\!\!\!\!\sum_{e \in \delta^-(V_z)}\!\!\!\! \nu_e \leq \!\!\!\sum_{e =vw \in \delta^-(V_{z})} \!\!\!\!\!\!\nu_e \ell'_w = \!\!\!\sum_{e \in \delta^-(V_{z})}\!\!\! x'_e = \!\!\!\sum_{e \in \delta^+(V_{z})} \!\!\!x'_e = \!\!\!\sum_{e=vw \in \delta^+(V_{z})} \!\!\!\!\!\!\nu_e \ell'_w \leq z_2\!\!\!\!\sum_{e \in \delta^+(V_{z})}\!\!\!\! \nu_e. 
    \]
    Thus
    \[ z_2 - z_1 = z_2(1 - \tfrac{z_1}{z_2}) \geq z_2\left(1 - \frac{\sum_{e \in \delta^+(V_z)} \nu_e}{\sum_{e \in \delta^-(V_z)} \nu_e}\right).%
              \]
              Since $z_2 > z_1$, $\sum_{e \in \delta^-(V_{z})} \nu_e \leq M + u_0 \leq 2M$, and $z_2 \geq \ell'_s(\theta) = 1$,
              $z_2 - z_1 \geq 1/(2M)$.

    Our next goal is to bound the queue lengths.
    The only extra ingredient we need is a bound on the speed at which a queue can grow. 
    There is always a {\sc ntfr} which does not route flow along cycles \citep[see][Theorem 6.64]{Koch}, so that $x'_e\le u_0$ for all $e \in \Ei$. 
    Therefore, $\max_v \ell'_v = \max\{\ell'_s, \max_e x'_e / \nu_e\} \leq u_0$.
Hence,
for all times $\theta$ before steady state is reached,
\[ 
    \ell_t(\theta) - \ell_s(\theta) \leq \ell_t(0) - \ell_s(0) + (u_0 - 1) \theta \leq T + 2(u_0-1) M^2T \leq 2u_0M^2T.
\]
This implies the same bound on all queue waiting times.
\end{proof}

\section{Some constructions and conjectures}\label{examples}

While we have settled the finite-time convergence to a steady state, there are a number of questions 
about dynamic equilibria that remain open. In this section we provide some constructions exhibiting somehow surprising behavior. First we show that in a dynamic equilibrium the flow across a cut can be arbitrarily larger than the inflow. Then we build an instance for which the dynamic equilibria has exponentially many phases. We wrap up the section by discussing the possibility of characterizing the steady state queues, and some conjectures regarding more general steady state results. 

\subsection{Flow across a cut}

As mentioned in the introduction a first conjecture would be that, similarly to what happens for static flows,  the flow across any 
cut is always bounded by the inflow. This would provide a way to estimate the queues and to prove their boundedness.   
Unfortunately  the property fails in a dynamic equilibrium. The reason for this is that flow entering the network at different 
times may experience different delays in such a way that they later superpose across an intermediate cut. 
In Example~\ref{eg:1} we constructed one such instance, where the peak outflow rate was a factor $13/12$ larger than the inflow rate. 
We will now show how to construct instances in which the outflow is arbitrarily larger than the inflow. 

\begin{figure}[t]
    \centering

    \begin{tikzpicture}[->,>=stealth',shorten >=1pt,auto,node distance=2.8cm, thick, scale=0.6, every node/.style={scale=0.6}]
  \tikzstyle{every state}=[fill=black!20,draw=none,text=black]

  \node[state] (s)   at (0,0)                 {\large $s$};
  \node[state]         (t) at (7,0) {\large $t$};
  \node[state]         (v) at (3,-3) {\large $v$};
    \node[state]         (t') at (14,0) {\large $t'$};
    
  \path (s) edge   [bend left=50]   node [pos=0.1, above, sloped]  {\large $\nu_e=u/3$} 
   						  node [pos=0.5, above, sloped]  {\large $\tau_e=5\plen/3$} (t)
                  edge   [bend right] 	  node  [pos=0.2, below, sloped]  {\large $\nu_f=3u/4$} 
            					  node  [pos=0.68, below, sloped]  {\large $\tau_f=0$}  (v)
            (v) edge   [bend left] 	  node  [pos=0.18, above, sloped]  {\large $\nu_g=u/3$}
        						  node  [pos=0.6, above, sloped]  {\large $\tau_g=0$}  (t)
                 edge   [bend right]        node  [pos=0.13, below, sloped]  {\large $\nu_h=u$}
            					  node  [pos=0.6, below, sloped]  {\large $\tau_h=5\plen/3$} (t)
            (t) edge   [bend left=50]   node [pos=0.12, above, sloped]  {\large $\nu_{e'}=u$} 
   						  node [pos=0.5, above, sloped]  {\large $\tau_{e'}=\plen/4$} (t')
                edge   [bend right=50]   node [pos=0.14, below, sloped]  {\large $\nu_{f'}=u/3$} 
   						  node [pos=0.5, below, sloped]  {\large $\tau_{f'}=0$} (t');
\end{tikzpicture}
\caption{The modified instance.}\label{fig:chain}
\end{figure}
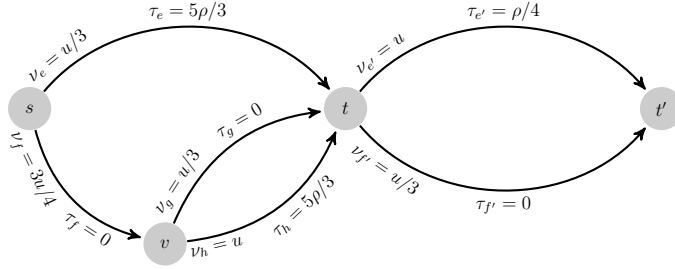

The base of the construction, given in Figure~\ref{fig:chain}, consists on appending to the instance in Example~\ref{eg:1} (slightly rescaled for convenience) an extra sink $t'$ and two arcs, $e'$ and $f'$, connecting $t$ with $t'$. Since the outflow in Example~\ref{eg:1} (with $\tau = 5\plen/6$) is
\begin{equation}\nonumber
\left\{\begin{array}{cl}
u/3&\mbox{ for }\theta\in [0,5\plen/4)\\
2u/3&\mbox{ for }\theta\in [5\plen/4,3\plen/2)\\
13u/12&\mbox{ for }\theta\in [3\plen/2,5\plen/2)\\
u&\mbox{ for }\theta\in [5\plen/2,\infty),
\end{array}\right.
\end{equation}
the outflow in the new instance will be 
\begin{equation}\nonumber
\left\{\begin{array}{cl}
u/3&\mbox{ for }\theta\in [0,7\plen/4)\\
13u/12&\mbox{ for }\theta\in [7\plen/4,11\plen/4)\\
u&\mbox{ for }\theta\in [11\plen/4,\infty).\\
\end{array}\right.
\end{equation}
Indeed, while the flow leaving $t$ is $u/3$, it will directly take arc $f'$ and therefore immediately reach $t'$. 
Then the flow leaving $t$ increases to $2u/3$ and it will continue to choose arc $f'$, though a queue will start to build up on that arc. 
At time $3\plen/2$ (considering the local time at $t$) the queue in $f'$ will be of size $(u/3)\cdot(\plen/4)$, implying a queuing time of $\plen/4$. 
Since this is exactly the delay of arc $e'$, at this point $e'$ enters the shortest path network. 
From this point the flow leaving $t$ increases to $13u/12$ and thus the flow splits: 
$u/3$ flow units take arc $f'$ while the remaining $9u/12$ take arc $e'$. 
The latter pattern stays until time $5\plen/2$ when the flow leaving $t$ changes to $u$ and thus it splits as $u/3$ taking arc $f'$ and  $2u/3$ taking arc $e'$, at which point steady state has been reached.

Some observations are in order: 
\begin{itemize}[-]
    \item The length of the ``pulse'' in this construction is exactly $\plen$, and so this can be made as large as required.
    \item The full pulse is produced as long as the inflow is $u$ for a period of $7\plen/4$; if it were to decrease or otherwise vary after this, it would only interfere with the final steady state phase.
    \item All arcs in the gadget have capacity at least $u/3$, meaning that no queues will form if the inflow is bounded by $u/3$.
        This means that the instance would still produces a pulse if the inflow was at most $u/3$ for some initial period, and then equal to $u$ for a period (of length at least $7\plen/4$).
\end{itemize}

Building on these observations, 
we now construct an instance where the ``pulse'' is arbitrarily larger than the inflow.
More precisely, we will construct a gadget $\pulse(u, k, q)$, for any $k \in \mathbb{N}$ and $u, q \in \mathbb{R}_+$, with the following properties (we define $\lambda := 13/12$ for convenience).
\begin{enumerate}[(i)]
\item\label{prop:flow} 
    There is some value $\warmupc_k = O((5/3)^k)$ such that, assuming a constant inflow rate of $u$ in the interval $[0, (5/3)^k\plen]$, the inflow rate into the sink is at most $u/3$ for $\theta < \plen\warmupc_k$ and exactly $\lambda^k \cdot u$ for $\theta \in [\plen\warmupc_k, \plen\warmupc_k + \plen)$ (where $\theta$ is the local time at the sink).
\item\label{prop:mincap} All arcs have capacity at least $u/3$.
\item\label{prop:size} The gadget has $6k$ arcs, with all free-flow delays being multiples of $3^{-k}\plen$ and bounded by $O((5/3)^k\plen)$, and all arc capacities being multiples of $12^{-k}u$ and bounded by $\lambda^k u$. 
\end{enumerate}

We have already seen in Figure~\ref{fig:chain} the construction of a $\pulse(u, 1, \plen)$ gadget; all the required properties clearly hold.
We construct a $\pulse(u, k, \plen)$ gadget for $k \geq 1$ by chaining together in series a $\pulse(u, k-1, \tfrac53\plen)$ gadget (call it $G$) followed by a $\pulse(\lambda^{k-1}u, 1, \plen)$ gadget (call it $H$).
Properties (\ref{prop:flow})--(\ref{prop:size}) are then mostly easy to verify inductively.
An important observation is that since the outflow of $G$ before its pulse is at most $u/3$, which is smaller than the smallest capacity arc in $H$, this initial flow does not cause any disturbance.
The bound on $\warmupc_k$ follows easily by making the stronger inductive claim that $\warmupc_k = \tfrac{21}{8}((5/3)^k - 1)$.
Now $\warmupc_1 = 7/4$, so this is correct for $k=1$.
For $k \geq 2$,
the flow into $H$ increases to $\lambda^{k-1}u$ at time $\plen\warmupc_{k-1}\cdot (5/3)$, meaning that the outflow increases to $\lambda^k u$ at time $\plen\warmupc_{k-1}(5/3) + (7/4)\plen = \plen\warmupc_k$, and this the inductive claim holds.

\subsection{Instances with an exponential number of phases}
\label{sec:exponential}
A natural hope would be that the number of phases is always polynomial in the input size (ideally as measured by the number of arcs in the instance, but failing that, as measured by the total encoding length of the instance). Unfortunately, as we show next, this is not the case and indeed the number of phases of a dynamic equilibrium may be exponential even is relatively simple series-parallel networks. This may help explaining why it is so notoriously difficult to practically compute dynamic equilibria in real-world networks \citep{W12,FH18}.

The $\pulse$ gadget of the previous section will be a first key ingredient.
The next step of our construction is to use it to build a ``damping'' gadget.
For any $k \in \mathbb{Z}_+$, $\qon \in \mathbb{R}_+$ we construct a gadget $\damper(k, \qon)$ 
with the following properties (recall that $\lambda := 13/12$).
\begin{enumerate}[(i)]
    \item There are values $\dampreq = e^{O(k)}\qon$ and $\theta_1 +2\qon < \theta_2 = e^{O(k)}\qon$ so that the following holds.
        Given an inflow of rate 1 in the interval $[0, \dampreq)$, the gadget produces an outflow that is precisely $1$ in the intervals $[\theta_1, \theta_1 + \qon)$ and $[\theta_2, \theta_2 + \qon)$, and precisely $\lambda^{-k}$ in the interval $[\theta_2 - \qoff, \theta_2)$.
    \item All arc capacities are multiples of $12^{-k}$ between $\lambda^{-k}/3$ and $1$.
    \item\label{prop:damplength} Assuming an inflow rate that is always bounded by $1$, the sum of queueing delays and free-flow delays within the gadget can never exceed $e^{O(k)}\qon$.
    \item The construction has $O(k)$ arcs, with total encoding length $O(k\cdot |\qon|)$ (where $|\qon|$ denotes the encoding length of $\qon$).
\end{enumerate}

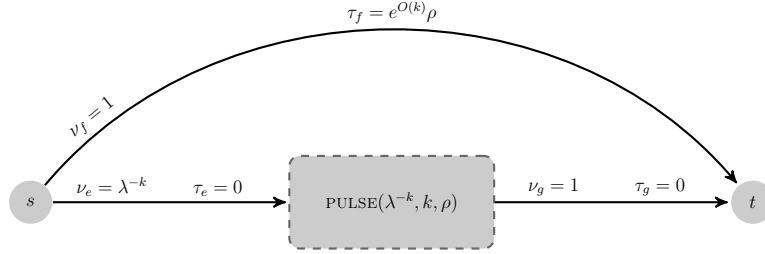
\begin{figure}
    \centering

    \begin{tikzpicture}[->,>=stealth',shorten >=1pt,auto,node distance=2.8cm, thick, scale=0.6, every node/.style={scale=0.6}]
  \tikzstyle{every state}=[fill=black!20,draw=none,text=black]

  \node[state] (s)   at (0,0)                 {\large $s$};
  \node[gadget] (pulse) at (8,0) {\large $\pulse(\lambda^{-k}, k, \qon)$};
    \node[state]         (t) at (16,0) {\large $t$};
    
    \path 
            (s) edge   [bend left=50]   node [pos=0.1, above, sloped]  {\large $\nu_f=1$} 
    node [pos=0.5, above, sloped]  {\large $\tau_f=e^{O(k)}\plen$} (t)
                                                  (s) edge 	  node  [pos=0.25, above, sloped]  {\large $\nu_e=\lambda^{-k}$}
        						  node  [pos=0.7, above, sloped]  {\large $\tau_e=0$}  (pulse)
            (pulse)
                 edge        node  [pos=0.25, above, sloped]  {\large $\nu_g=1$}
            					  node  [pos=0.7, above, sloped]  {\large $\tau_g=0$} (t);
\end{tikzpicture}

    \caption{Construction of a $\damper(k, \qon)$ gadget.}
    \label{fig:damper}
\end{figure}

The construction of this gadget is shown in Figure~\ref{fig:damper} (the precise value of $\tau_f$ is discussed below).
Initially, no flow uses arc $f$; a queue grows on $e$, which sends the correct inflow of $\lambda^{-k}$ into the $\pulse(\lambda^{-k}, k,\qon)$ gadget. 
After some time, the pulse gadget generates a pulse of size exactly 1, for a period of length $\qon$.
We set $\theta_1$ to the time that this pulse begins (as measured at $t$); by Property~(\ref{prop:flow}) of the $\pulse$ gadget, $\theta_1 = O((5/3)^k\qon)$.
Since all free-flow delays in the gadget are bounded by $O((5/3)^k\qon)$, and all capacities are multiples of $12^{-k}$ bounded by $1$,
Theorem~\ref{thm:pseudo}
yields a bound of $e^{O(k)}\qon$,
both on the time to reach steady state, and also on the total delay within the pulse gadget.
Once the pulse gadget has reached steady state, its outflow remains at $\lambda^{-k}$; this will be the outflow of the $\damper$ gadget as well, as long as $f$ has not joined the shortest path network.

We now see how to choose $\tau_f$; large enough that $f$ only joins the shortest path network after the pulse gadget has been sending outflow $\lambda^{-k}$ for at least $\qoff$ amount of time.
Since the queue on arc $e$ grows at rate $\lambda^k-1$, and the delay within the pulse gadget itself is at most $e^{O(k)}\qon$, 
it is clear that we can choose $\tau_f = e^{O(k)}\qon$.
Once $f$ does join the shortest path network, the entire $\damper$ gadget reaches steady state, and the flow into $t$ increases to $1$.
This determines $\theta_2$, and from this we can fix $Q$, 
ensuring that both are bounded by $e^{O(k)}\qon$. 
Property~(\ref{prop:damplength}) also follows immediately from the choice of $\tau_f$.

Now we come to the construction of the gadget $\exponential(d)$, which for any $d \in \mathbb{Z}_+$ will have size quadratic in $d$, and at least $2^d$ phases.
The construction is recursive.
In the following, $C$ will denote a constant chosen large enough in relation to the hidden implicit constants in the definition of the $\damper$ gadget; the precise requirements on $C$ will become clear.
We construct $\exponential(1)$ by simply taking two parallel arcs, one of capacity $1/3$ and length $0$, and the other of capacity $2/3$ and length 1; this clearly has two phases.
To construct $\exponential(d)$ for $d \geq 2$, 
take a $\damper(15d, C^{(d-1)^2})$ gadget (call it $G$), 
and follow this in series by an $\exponential(d-1)$ gadget (call it $H$).

The idea behind this construction is that because the outflow $\lambda^{-15d}$ from the damper gadget $G$ during the damped period is smaller than the minimum arc capacity (of at least $\lambda^{-15(d-1)}/3$) in $H$, queues within the gadget will decrease.
The length of the damped phase has been chosen to be long enough that all queues in $H$ empty out completely (this is the only really delicate aspect of this construction).
The two high outflow periods of $G$ last long enough that $H$ runs (inductively) through $2^{d-1}$ phases during both periods, giving a total of at least $2^d$ phases.

The following properties about $\exponential(d)$ are then straightforward to confirm inductively, exploiting also the properties of the $\damper$ gadget.
\begin{enumerate}[(i)]
    \item Given a constant inflow of 1 in the interval $[0,C^{d^2})$, the gadget goes through $2^d$ phases.
    \item\label{prop:expcaps} All arc capacities are multiples of $12^{-15d}$ between $\lambda^{-15d}/3$ and $1$.     \item\label{prop:explength} Assuming an inflow rate that is always bounded by $1$, the sum of queueing delays and free-flow delays within the gadget can never exceed $C^{d^2 - d}$. 
    \item The gadget consists of less than $(10d)^2$ arcs, and has encoding length $O(d^4)$. 
\end{enumerate}
Property~(\ref{prop:explength}) requires that $C$ is chosen large enough that the bound on the total delays in $G$ given by 
Property~(\ref{prop:damplength}) of the \damper{} gadget is smaller than $C^{k/15 - 2}\plen$, which for $k=15d$ and $\plen=C^{(d-1)^2}$ yields a bound of $C^{d^2 - d - 1}$. 
Inductively the total delays in $H$ sum to at most $C^{(d-1)^2 - (d-1)} = C^{d^2 - 3d}$.
Overall, we obtain a bound of at most $C^{d^2-d-1} + C^{d^2-3d} \leq C^{d^2-d}$.

Let us now see that all queues in $H$ do empty out during the damped period.
By Properties~(\ref{prop:expcaps}) and (\ref{prop:explength}), Theorem~\ref{thm:pseudo} 
tells us that the time to reach steady state from the beginning of the damped period 
is at most 
\[ 2C^{(d-1)^2 - (d-1)}\cdot (10d)^4\cdot (12^{15d})^2 \leq C^{d^2 - 3d + 2}\cdot C^{d-1} = C^{(d-1)^2}, \]
which is the length of the damped phase.
(We assume in the above that $C$ is large enough that $2(10d)^4\cdot(12^{15d})^2 \leq C^{d-1}$.)
Since the inflow into $H$ in the damped period, namely $\lambda^{-15d} < \lambda^{-15(d-1)}/3$, is lower than the capacity of any arc in $H$, the steady state necessarily has no queues.

\medskip

While we know that the number of phases may be very large, 
it is natural to expect that there are only a \emph{finite} number of phases.
While we conjecture that this is true, it is not ruled out by our result.
Our result does show that if the length of all phases in the evolution is bounded away from zero, 
then there can only be a finite number.
It is not ruled out, however, that an infinite sequence of phases occurs in a finite amount of time.
This is the same issue discussed in \citet{ccl}.
The issue is significant; 
it is the one obstacle to showing uniqueness (in an appropriate sense) of dynamic equilibria.

If such a result could be shown, an even stronger conjecture would be that the number of phases is 
pseudopolynomially bounded in the input size.
This would show that the exponential capacities and free-flow delays in the $\exponential$ gadget construction are in fact necessary.

\subsection{Steady state queue lengths}\label{examples-ssq}

Knowing that the dynamic equilibrium always reaches a steady state, a natural question is whether  steady state queues can 
be characterized without having to compute the full equilibrium evolution. While we already observe that this is the case 
when the dual problem \eqref{eq:dual} has a unique solution, which occurs generically, the following example suggests that this is likely not possible in general.
\begin{example}\label{eg:3}
Consider the network of Example~\ref{eg:1}, setting $\tau=2$ and $u=1$, with an extra node $\hat t$, which becomes the new sink, and two additional arcs, $a=(t,\hat t)$ and $b=(t,\hat t)$. Let $\nu_a=2/3$, $\nu_b=1/3$, $\tau_a=0$, and $\tau_b=1$. 
\begin{figure}[h]
    \centering
\begin{tikzpicture}[->,>=stealth',shorten >=1pt,auto,node distance=2.8cm, thick, scale=0.6, every node/.style={scale=0.6}]
\tikzstyle{every state}=[fill=black!20,draw=none,text=black]

  \node[state] (s)   at (0,0)                 {\large $s$};
  \node[state]         (t) at (7,0) {\large $t$};
  \node[state]         (t') at (12,0) {\large $\hat t$};
  \node[state]         (v) at (3,-3) {\large $v$};

  \path (s) edge   [bend left=50]   node [pos=0.12, above, sloped]  {\large $\nu_e=1/3$} 
   						  node [pos=0.5, above, sloped]  {\large $\tau_e=2$} (t)
                  edge   [bend right] 	  node  [pos=0.18, below, sloped]  {\large $\nu_f=3/4$} 
            					  node  [pos=0.65, below, sloped]  {\large $\tau_f=0$}  (v)
            (v) edge   [bend left] 	  node  [pos=0.18, above, sloped]  {\large $\nu_g=1/3$}
        						  node  [pos=0.6, above, sloped]  {\large $\tau_g=0$}  (t)
                 edge   [bend right]        node  [pos=0.15, below, sloped]  {\large $\nu_h=1$}
            					  node  [pos=0.6, below, sloped]  {\large $\tau_h=2$} (t)
	  (t) edge   [bend left=50]   node [pos=0.18, above, sloped]  {\large $\nu_a=2/3$} 
   						  node [pos=0.6, above, sloped]  {\large $\tau_a=0$} (t')
	 (t) edge   [bend right=50]   node [pos=0.18, below, sloped]  {\large $\nu_b=1/3$} 
   						  node [pos=0.6, below, sloped]  {\large $\tau_b=1$} (t');
\end{tikzpicture}
\end{figure}
Clearly, up to time $3+3/5$ all flow will simply take arc $a$ and will not queue at $t$. 
Therefore we can ignore this initial phase, and the queues that will form at equilibrium in arcs $a$ and $b$ are the same as those that we would have in a network consisting of just nodes $t$ (the source) and $\hat t$ (the sink) and inflow
\begin{equation}\nonumber
u_0(\theta)=\left\{\begin{array}{cl}
13/12&\mbox{ for }\theta\in [0,2+2/5)\\
1&\mbox{ for }\theta\in [2+2/5,\infty).\\
\end{array}\right.
\end{equation}
In this instance all flow will take arc $a$ for time $\theta\in [0,8/5)$, forming a queue $z_e(8/5)=2/3$. At this point flow will start splitting between arcs $a$ and $b$ in proportions $2/3$, $1/3$, implying that queues will grow on both arcs until time $2+2/5$ where the steady state is achieved. The steady state queues will thus be $z^*_a=32/45$ and $z^*_b=1/45$. This example shows that the steady state queues are not minimal in any reasonable sense and that, furthermore, slightly changing the instance (e.g. $\tau_4$) will change the steady state queues. Furthermore, if we slightly increase the capacity of arc $b$, say to $1/3+\varepsilon$ the steady state queues jump to $z^*_a=2/3$ and $z^*_b=0$.  

Additionally, one can observe from a slight variant of this instance, namely taking $\tau$ large and $\nu_b=1/3+\varepsilon$, that queues may grow very large in the transient and then go down to zero at steady state.
\end{example}

\subsection{Conjectures on more general steady state results}

Suppose that the inflow into the network is not a constant, but a time varying function $u_0(\theta)$.
Suppose moreover that $u_0(\theta)$ is always bounded by the min-cut capacity of the network.
Of course, there cannot be convergence to a steady state in this setting;
but it is natural to expect that queues stay bounded.
We conjecture that this is the case.
It is not clear how our potential argument can aid in proving this conjecture.
In particular, note that the boundedness of $\Phi$ alone is not helpful, as this does not imply any bounds on the queue sizes.

\medskip

Suppose now that the inflow is constant, but \emph{larger} than the min-cut capacity.
Then again the evolution can of course not converge to a steady state in the way that we have defined it:
queues cannot remain bounded.
However, we conjecture that it is still true that after a finite amount of time, the equilibrium settles into a final phase that lasts forever.\footnote{Subsequent to this work, this conjecture has been resolved positively. See: Olver, Sering and Vargas Koch. Continuity, Uniqueness and Long-Term Behavior of Nash Flows Over Time, 2021. Available online at \href{https://arxiv.org/abs/2111.06877}{arXiv:2111.06877}.}

\medskip

\paragraph*{Acknowledgements.}
We thank the organizers of the Dagstuhl Seminar \emph{``Dynamic Traffic Models in Transportation Science''} held in 2015, where this work was initiated. We sincerely thank Vincent Acary, Umang Bhaskar and Martin Skutella for enlightening discussions, and Jonas Isreal for helpful comments on Theorem~\ref{thm:pseudo}. We also thank the Associate Editor at \emph{Operations Research} and two reviewers for their comments that helped improve the presentation.

This work was partially supported by ANID through grants FONDECYT 1190043, FONDECYT 1171501, Basal  AFB-170001 and Basal AFB-180003; and by the Dutch Science Foundation (NWO) through a TOP grant (614.001.510) and a Vidi grant (016.Vidi.189.087).

\end{document}